\documentclass{article}

\usepackage{amsmath, amssymb, amsthm}
\usepackage{booktabs}
\usepackage{makecell}
\usepackage{longtable}
\usepackage{hyperref}
\usepackage{graphicx}
\usepackage{authblk}
\usepackage{rotating}

\theoremstyle{plain}
\newtheorem{theorem}{Theorem}[section]
\newtheorem{proposition}[theorem]{Proposition}
\newtheorem{lemma}[theorem]{Lemma}

\theoremstyle{definition}
\newtheorem{definition}[theorem]{Definition}
\newtheorem{example}[theorem]{Example}

\theoremstyle{remark}

\newcommand{\GF}{\mathbb{F}}
\newcommand{\ZZ}{\mathbb{Z}}

\begin{document}

\title{Multiplicative Modular Nim (MuM)}

\author{Satyam Tyagi\thanks{Email: \href{mailto:satyam.tyagi@gmail.com}{satyam.tyagi@gmail.com}, ORCID: \href{https://orcid.org/0009-0001-0137-2338}{0009-0001-0137-2338}}}
\affil{Independent Researcher}

\abstract{\par We introduce Multiplicative Modular Nim (MuM), a variant of Nim in which the 
  traditional nim-sum is replaced by heap-size multiplication modulo~$m$. We 
  establish a complete theory for this game, beginning with a direct, Bouton-style 
  analysis for prime moduli. Our central result is an analogue of the 
  Sprague-Grundy theorem, where we define a game-theoretic value, the \emph{mumber}, 
  for each position via a multiplicative \texttt{mex} recursion. We prove that 
  these mumbers are equivalent to the heap-product modulo~$m$, and show that for 
  disjunctive sums of games, they combine via modular multiplication in contrast to 
  the XOR-sum of classical nimbers.
  
\par For composite moduli, we show that MuM decomposes via the Chinese Remainder Theorem 
  into independent subgames corresponding to its prime-power factors. We extend the game 
  to finite fields~$\GF(p^n)$, motivated by the pedagogical need to make the algebra of the 
  AES S-box more accessible. We demonstrate that a sound game in this domain requires a 
  \emph{Canonical Heap Model} to resolve the many-to-one mapping from integer heaps 
  to field elements. To our knowledge, this is the first systematic analysis of a 
  multiplicative modular variant of Nim and its extension into a complete, 
  non-additive combinatorial game algebra.}

\bigskip
\noindent\textbf{Keywords:} Combinatorial Game Theory, Impartial Games, Nim, Modular Arithmetic, Finite Fields

\bigskip
\noindent\textbf{MSC Classification:} 91A46, 11A07

\bigskip
\noindent\textbf{ACM Classification:} G.2.1 Combinatorics

\maketitle


\section{Introduction}

Classical Nim is a cornerstone of combinatorial game theory, with a simple structure 
and a deep strategy grounded in the binary XOR operation~\cite{Bouton1901}. In this 
work, we explore a multiplicative analogue of Nim where the game state evolves 
through modular multiplication rather than the usual nim-sum. This seemingly simple 
change leads to a rich new game, Multiplicative Modular Nim (MuM), that introduces 
players to field-like behavior in a playful setting.

Our analysis proceeds in stages of increasing generality. We begin with a direct, Bouton-style analysis to establish the winning strategy for prime moduli. We then build a more general, Sprague-Grundy-style theory that holds for any modulus~$m$. We define a \emph{mumber} for each position using a multiplicative \texttt{mex} recursion and prove constructively that this value is equivalent to the product of heaps modulo~$m$. This demonstrates that the modular-product rule is not merely an imposed condition, but an emergent property of the game's underlying combinatorial structure.

We then apply this general theory to analyze the structure of games with composite moduli, showing how the Chinese Remainder Theorem allows for their decomposition into independent subgames corresponding to their prime-power factors. This framework naturally leads to our original motivation: creating a pedagogical bridge to the algebra of finite fields, particularly the field $\GF(2^{8})$ used in the AES S-box. We define a final extension, \emph{Polynomial MuM}, and show that a sound game in this domain requires a \emph{Canonical Heap Model} to resolve the complex mapping between integer heaps and field elements. The resulting game provides a playable model whose optimal strategy requires precisely the field-multiplication and inversion operations that underpin modern cryptography.

\section{Literature Review}\label{literature-review}

\subsection{"Mod-m" Addition Variants}
Work on Nim-like games with modular invariants has a rich history, yet it has focused exclusively on additive structures.
\begin{itemize}
    \item \cite{Kotzig1946} explored games where heap sizes are reduced modulo~$m$ by subtraction. This was extended by Fraenkel, Jaffray, Kotzig \& Sabidussi, who coined the term "Modular Nim" and proved the existence of periodic Sprague-Grundy patterns~\cite{Fraenkel1995}.
    \item More recent work by \cite{Tan2014} and \cite{Horrocks2019} has established further properties of these additive modular games, such as periodicity theorems and explicit computation of nim-values for various move-sets.
\end{itemize}
All of these foundational papers explore invariants based on modular addition or subtraction; none use a product-based invariant.

\subsection{Multiplicative Structures in Impartial Games}
Multiplication has appeared in combinatorial game theory, but typically as a secondary structure built upon the primary additive (nim-sum) framework.
\begin{itemize}
    \item Lenstra, in foundational work, defined \emph{nim-multiplication}, proving that the set of nimbers forms a field under nim-sum and this new product~\cite{Lenstra1977}. This established a precedent for embedding field structures into game analysis.
    \item Conway's analysis of "Turning Corners" uses the nim-product of coordinates to determine the nim-value of a position, but the winning strategy still relies on computing the nim-sum of these products~\cite{Conway2003}.
    \item Other games feature multiplication in their move sets, such as replacing a number with one of its proper factors~\cite{Cipra2018}, but not as the core invariant determining the game's outcome.
\end{itemize}
To our knowledge, no prior work uses modular multiplication as the primary winning invariant, nor does it develop a corresponding non-additive game algebra. MuM fills this previously unstudied niche at the intersection of modular arithmetic and impartial game theory.

\section{Multiplicative Modular Nim (MuM)}\label{sec:game-def}

\subsection{Game Definition for Prime Moduli}\label{game-setup}

We first define the game for a prime modulus~$p$, where the algebraic structure is a field.

\begin{definition}[MuM for Prime Moduli]
\leavevmode
\begin{itemize}
  \item \textbf{Game Play:} The game is played with $n$ heaps, each containing a positive integer value. A fixed prime number $p$ is chosen at the start.
  \item \textbf{Heap Constraint:} No heap is allowed to hold a value divisible by $p$ (i.e., for any heap $h_i$, $h_i \not\equiv 0 \pmod p$).
  \item \textbf{Legal Move:} A player chooses any heap $h_i > 1$ and replaces it with $h_i' = h_i - r$, where the integer $r$ satisfies:
    \begin{enumerate}
      \item $1 \le r < p$ (the amount subtracted is less than the modulus);
      \item $r < h_i$ (the heap strictly decreases);
      \item The new heap $h_i'$ is not divisible by $p$.
    \end{enumerate}
\end{itemize}
\end{definition}

\subsection{A Bouton-Style Analysis for Prime Moduli}\label{bouton-style-analysis}

The win/loss condition is determined by the product of the heaps modulo~$p$. A position $H$ with product $P = \prod h_i$ is a \textbf{losing position} if $P \equiv 1 \pmod p$ and a \textbf{winning position} if $P \not\equiv 1 \pmod p$. The following propositions establish a complete winning strategy.

\begin{proposition}[Winning to Losing]\label{prop:winning}
If a position has product $P \not\equiv 1 \pmod p$ and at least one heap satisfies $h_j > p$, then there exists a legal move to a losing position.
\end{proposition}

\begin{proof}
Let $P = h_j \cdot m$ where $m = \prod_{i\neq j} h_i$. 
Because all heaps are coprime to $p$, so is $m$. 
Thus, its inverse $m^{-1} \pmod p$ exists. 
The player chooses $r \equiv h_j - m^{-1} \pmod p$. 
Since $h_j > p > m^{-1}$, the move $h_j \to h_j - r$ is a legal reduction. 
The new product is $P' = (h_j-r)m \equiv m^{-1}m \equiv 1 \pmod p$.
\end{proof}

\begin{proposition}[Losing to Winning]\label{prop:losing}
If a position has product $P \equiv 1 \pmod p$ and is non-terminal, every legal move leads to a winning position.
\end{proposition}

\begin{proof}
Let $m = \prod_{i\neq j} h_i$. 
The initial state has $h_j m \equiv 1 \pmod p$. 
A move $h_j \to h_j-r$ results in a new product $P' = (h_j-r)m$. 
If $P' \equiv 1 \pmod p$, then by subtracting the congruences, 
we get $-rm \equiv 0 \pmod p$. 
Since $m$ is invertible, this implies $r \equiv 0 \pmod p$, 
which contradicts the move rule $1 \le r < p$.
\end{proof}

\subsection{Stranded Positions and Consolidation}\label{product-consolidation}

When all heaps are smaller than $p$, a winning position may be \emph{stranded}, meaning no legal move can reach a losing state. For example, with $p=5$, the position $[2,2,2]$ has product $P \equiv 3 \pmod 5$ (winning), but every move results in a product of $4 \pmod 5$. To ensure termination, we apply the following rule:

\begin{definition}[Product Consolidation]
If a position is winning and stranded, the player on turn replaces the heap multiset $\{h_1, \dots, h_n\}$ with a single new heap $H_{\text{new}} = \prod h_i$.
\end{definition}

\begin{example}[Worked example (mod 5)]\leavevmode\par
\begin{enumerate}
    \item \([6,6,6]\): \(P = 216 \equiv 1 \pmod{5}\) — \emph{losing}.
    \item Player 1 plays to \([6,6,2]\), \(P = 72 \equiv 2 \pmod{5}\) — winning.
    \item Player 2 plays to \([6,3,2]\), \(P = 36 \equiv 1 \pmod{5}\) — losing.
    \item ... After some moves, the position is \([2,2,2]\), with \(P = 8 \equiv 3 \pmod{5}\). This is winning but \emph{stranded}.
    \item \emph{Consolidate:}\; The player replaces \([2,2,2]\) with a single heap of value \(8\).
    \item The player with heap \(8\) subtracts \(r=2\) (legal as $1 \le 2 < 5$) to get a new heap of \(6\). The product is now $P'=6 \equiv 1 \pmod 5$. This is a losing position for the opponent.
\end{enumerate}
\end{example}

Note that the single-heap subtraction game \textsc{Sub}$(k)$—reduce heap by $1$ to $k$, 
and terminal position $0$—is completely solved: a heap of size $h$ is a 
$\mathcal P$-position iff $h \equiv 0 \pmod{k+1}$ \cite{Conway2001, Albert2019}. 
For the consolidated heap we use the translated variant with terminal $1$ obtained 
by the offset $h \mapsto h-1$.

\subsection{Choice of the Consolidation Move}\label{choice-of-consolidation-move}

The plain product $H_{\text{new}} = \prod h_i$ is the most natural choice for three primary reasons:
\begin{enumerate}
\item \textbf{Outcome Preservation:} It ensures the mumber of the new position ($H_{\text{new}} \bmod p$) is identical to the mumber of the old position ($P \bmod p$), so the win/loss status does not arbitrarily change.
\item \textbf{Algebraic Consistency:} This physical operation of multiplying heaps directly mirrors the abstract algebraic structure of the game, where the mumber of a position is the product of the individual heap mumbers.
\item \textbf{Preservation of Gameplay Cadence:} Using the full product creates a new, larger heap that must be reduced via the standard subtractive rules, preserving the game's dynamic and avoiding an abrupt "teleport" to a small, near-terminal heap value.
\end{enumerate}

\section{A Sprague-Grundy Theory for MuM}\label{sec:sg-theory}

We now develop a more general theory that builds the game's mechanics from first principles. This demonstrates that the modular-product rule is not an imposed condition, but an emergent property of the game's combinatorial structure.

\begin{theorem}[MuM Grundy-Mex Theorem]\label{thm:MuM-SG}
Fix a prime modulus $p \ge 2$. For every MuM position $G$, define its \emph{mumber} $\operatorname{MM}(G) \in \{0,1,\dots ,p-1\}$ recursively by
\[
  \operatorname{MM}(G) = \operatorname{mex}_{p}\! \bigl\{\operatorname{MM}(G') \mid G' \text{ is an option of }G\bigr\},
\]
where $\operatorname{mex}_{p}(S) := \min\{t \in \{0, \dots, p-1\} \mid t \notin S\}$.
Then for any position $G=(h_1, \dots, h_k)$:
\begin{enumerate}
    \item The mumber is equivalent to the heap product: $\operatorname{MM}(G) \equiv \prod h_i \pmod p$.
    \item $G$ is a P-position (previous-player win) $\Longleftrightarrow \operatorname{MM}(G)=1$.
    \item For any disjunctive sum $G = G_1 + G_2$, the mumbers combine via multiplication:
    \[
        \operatorname{MM}(G) = \operatorname{MM}(G_1) \cdot \operatorname{MM}(G_2) \pmod p.
    \]
\end{enumerate}
\end{theorem}

\subsection{Proof of the MuM Grundy-Mex Theorem}\label{subsec:proof-MuM-SG}
We prove part (1) by induction on the \emph{move depth}, $\deg(G)$, defined as the shortest path from $G$ to a consolidated single-heap game.

\begin{lemma}[Single-Heap Base Case, $\deg(G)=0$]\label{lem:single-heap}
For a single-heap position $H(h)$, $\operatorname{MM}(H(h)) \equiv h \pmod p$.
\end{lemma}
\begin{proof}
By induction on $h$. For $h=1$, the option set is empty, so $\operatorname{mex}_p(\varnothing)=1$. For $h>1$, the options are $H(h-r)$ for $1 \le r < \min(p,h)$. By the inner induction hypothesis, their mumbers are $\{h-1, h-2, \dots\}$, so the first missing value is $h \pmod p$.
\end{proof}

\begin{proposition}[Induction Step]\label{prop:induction}
Assume $\operatorname{MM}(G') \equiv P(G') \pmod p$ for all positions $G'$ with $\deg(G') \le k$. Then the equality holds for any $G$ with $\deg(G)=k+1$.
\end{proposition}
\begin{proof}[Proof (idea, then details)]
\emph{Idea.} Lowering heap $h_i$ by $r\in\{1,\dots,p-1\}$ multiplies the product by $(h_i-r)/h_i$. As $r$ ranges, these factors sweep through all possible changes, so the new product attains every value except the original. Hence $\operatorname{mex}_p$ returns the original product $P(G)\bmod p$.

\emph{Formal argument.} Let $G=(h_1,\dots,h_t)$ with $\deg(G)=k+1$. Let $S = \{\operatorname{MM}(G')\}$ be the set of mumbers of its options. A legal move on $h_i$ by subtracting $r$ gives a child $G'(r)$ with $\deg(G') \le k$. By the induction hypothesis, its mumber is:
\[
  \operatorname{MM}\bigl(G'(r)\bigr) \equiv P(G'(r)) \equiv (h_i-r) \prod_{j\neq i} h_j \pmod p.
\]
Let $\alpha = \prod_{j\neq i} h_j$. Since $\gcd(\alpha, p)=1$, the map $r \mapsto (h_i-r)\alpha \pmod p$ for $r \in \{1,\dots,p-1\}$ is a bijection to the set of residues $\ZZ_p \setminus \{h_i\alpha \pmod p\}$. Thus, the set of reachable mumbers $S$ is precisely $\{0,1,\dots,p-1\}\setminus\{P(G)\bmod p\}$, and its mex is $P(G)\bmod p$.
\end{proof}

\clearpage
\subsection{Illustrative Examples}\label{sec:illustrations}

\begin{figure}[htbp!]
  \centering
  \includegraphics[width=1\textwidth, keepaspectratio]{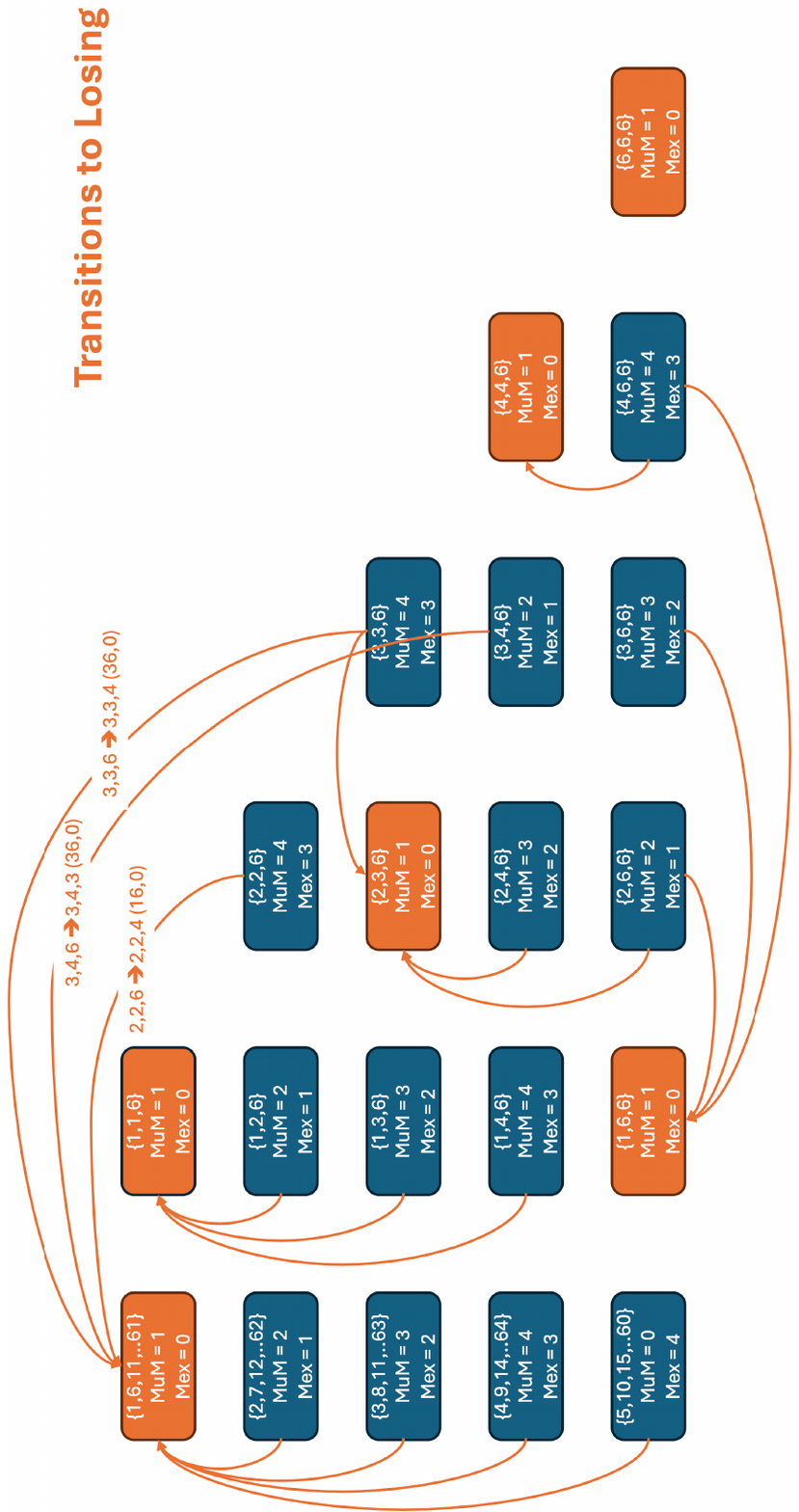}
  \caption{Transition graph to losing positions ($p=5$).}
  \label{fig:losing-mex}
\end{figure}

\clearpage
\begin{figure}[htbp!]
  \centering
  \includegraphics[width=1\textwidth, keepaspectratio]{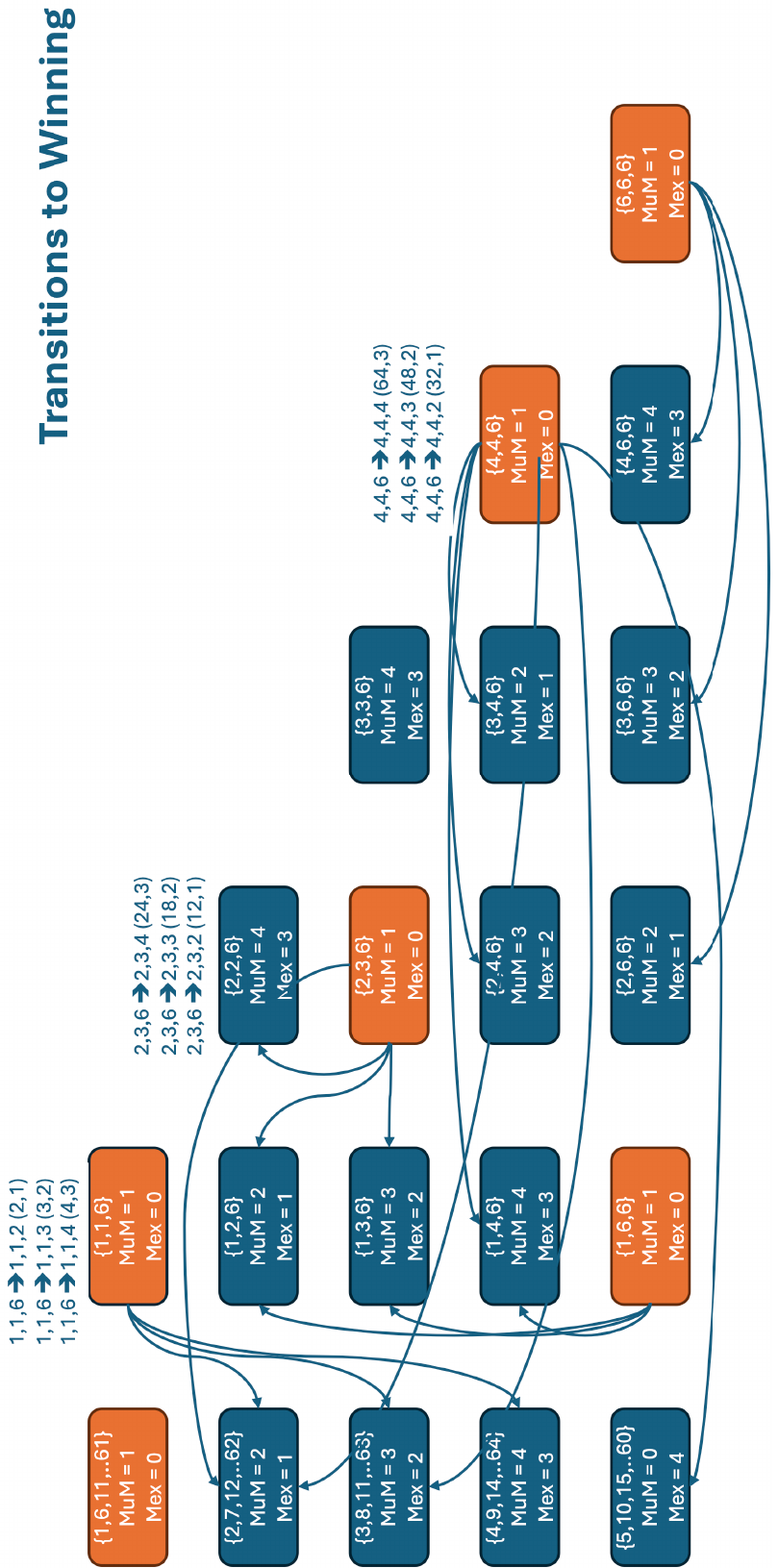}
  \caption{Transition graph to winning positions ($p=5$).}
  \label{fig:winning-mex}
\end{figure}

\clearpage
\begin{example}[Recursive Computation of Grundy-like Value $G(h)$ for $p=5$]
For a single heap, we compute the Grundy-like value $G(h) = \operatorname{mex}\{G(h-r) \mid 1 \le r < 5, h-r \ge 1 \}$.
\begin{center}
\begin{tabular}{c|l|c}
\toprule
\textbf{$h$} & \multicolumn{1}{c|}{\textbf{Successor Grundy Set}} & \textbf{$G(h)$} \\ \midrule
1 & $\varnothing$                        & 0 \\
2 & $\{G(1)=0\}$                     & 1 \\
3 & $\{G(1)=0, G(2)=1\}$            & 2 \\
4 & $\{G(1)=0, G(2)=1, G(3)=2\}$   & 3 \\
5 & $\{G(1)=0, G(2)=1, G(3)=2, G(4)=3\}$ & 4 \\
6 & $\{G(2)=1, G(3)=2, G(4)=3, G(5)=4\}$ & 0 \\
7 & $\{G(3)=2, G(4)=3, G(5)=4, G(6)=0\}$ & 1 \\
\bottomrule
\end{tabular}
\end{center}
Note that for $h \le p$, $G(h)=h-1$. For $h > p$, the pattern becomes periodic with period $p$, so $G(h) = (h-1)\bmod p$. This aligns with our theorem, as $\operatorname{MM}(h) = h \bmod 5$ for the mumber (with identity 1) and $G(h) = (h-1) \bmod 5$ for the Grundy value (with identity 0). See details in Appendix \ref{recursive-mex-table}.
\end{example}

\section{Analysis of Composite Moduli via CRT}\label{sec:crt-analysis}
The general theory holds for any modulus $m$, but for composite $m$, the Chinese Remainder Theorem (CRT) provides a powerful analytical tool.

\subsection{The Decomposed State Vector}
For a composite modulus $m$ with prime-power factors $\{m_1, \dots, m_N\}$, the state of a position $H$ can be analyzed via a \emph{decomposed state vector} $V(H) = (V_1, \dots, V_N)$, where $V_j \equiv \prod h_i \pmod{m_j}$. By the CRT, the losing condition $P \equiv 1 \pmod m$ is equivalent to the state vector being the identity, $V(H) = (1,1,\dots,1)$. See the worked out decomposition example of MuM\textsubscript{15} in Appendix~\ref{mod-3-mod-5-simplified-mumbers}.  

This isomorphism means playing MuM with a composite modulus $m$ is strategically equivalent to playing $N$ independent MuM games on the factor moduli $\{m_1, \dots, m_N\}$ simultaneously. For a square-free modulus like $m=30=2 \cdot 3 \cdot 5$, this results in a full decomposition into prime-modulus games. For a modulus like $m=90=2 \cdot 3^2 \cdot 5$, the maximal decomposition is into the games on moduli $\{2, 9, 5\}$.

\subsection{Analogy Table: Nim vs MuM}\label{analogy-table-nim-vs-mum}
\begin{table}[htbp!]
\centering
\caption{A comparison of the algebraic structures of Classical Nim and MuM.}
\begin{tabular}{@{}lll@{}}
\toprule
\textbf{Feature} & \textbf{Classical Nim (Nimbers)} & \textbf{MuM (Mumbers)} \\ \midrule
Operation & Additive (XOR, $\oplus$) & Multiplicative ($\cdot \bmod p$) \\
Identity & 0 & 1 \\
Game Value & Grundy Number ($g$) & Mumber ($MM$) \\
Sum of Games & $g(G_1) \oplus g(G_2)$ & $MM(G_1) \cdot MM(G_2) \pmod p$ \\
Structure & Vector space over $\GF(2)$ & Multiplicative group $\ZZ_p^{*}$ \\
Value Domain & Non-negative integers & Residues modulo $p$ \\
\bottomrule
\end{tabular}
\end{table}

\clearpage
\section{Extension to Polynomial Fields: A Model for \texorpdfstring{$\GF(p^{n})$}{GF(p\textasciicircum n)}}\label{sec:polyMum}

\subsection{Motivation: Returning to the Starting Point}
Our original motivation arose from a desire to build a *pedagogical bridge* to the algebra of finite fields---especially the field $\GF(2^{8})$ used in the AES S-box. We now show that the core principles of MuM extend from the ring $\ZZ_{m}$ to any finite field $\GF(p^{n})$.

\subsection{The Challenge of Polynomial Representation}\label{subsec:poly_challenge}
A direct translation reveals a structural challenge. In numerical MuM over $\ZZ_p$, the rule $1 \le r < p$ guarantees that $r \not\equiv 0 \pmod p$, ensuring any move $h \to h-r$ changes the residue. This simple guarantee breaks down in the polynomial world. The mapping from an integer heap $h$ to a field element $s(h) = \text{poly}(h) \pmod{I(x)}$ is many-to-one (e.g., in $\GF(2^{3})$, $h=9$ and $h'=2$ both map to $x$). A player could make a "synonym move" ($9 \to 2$) that reduces the integer but leaves the field element unchanged, allowing them to pass a losing turn. To create a sound game, we must refine the definition of the game's objects.

\subsection{The Canonical Heap Model for Poly-MuM}\label{subsec:polyMumDef}
We resolve the synonym problem by defining the game using unique representatives.
\begin{definition}[Polynomial MuM]
\leavevmode
\begin{itemize}
  \item \textbf{Field:} The game is played in $\GF(p^{n}) \cong \GF_{p}[x]/\langle I(x) \rangle$.
  \item \textbf{Canonical Heap:} For every non-zero field element $s$, its \emph{canonical integer representative}, $C(s)$, is the unique integer $h \in \{1, \dots, p^n-1\}$ corresponding to its standard polynomial representation. Heaps are a multiset of these canonical integers.
  \item \textbf{State Invariant:} The product $P_{\text{field}} = \prod_{h \in H} s(h)$. A position is \emph{losing} iff $P_{\text{field}} = 1$.
  \item \textbf{Legal Move:} A player chooses a heap $h$ and replaces it with a new canonical heap $h'$ such that $h' < h$. Because all heaps are canonical, this ensures $s(h') \neq s(h)$, making every move meaningful.
\end{itemize}
\end{definition}

\subsection{Game Theory of Poly-MuM}
The game theory directly parallels the numerical version.
\paragraph{Winning and Losing.} From a winning position ($P_{\text{field}} \neq 1$) with a heap $h \ge p^n-1$, a move to a losing state is guaranteed. The target heap is $k = C(s(h) \cdot P_{\text{field}}^{-1})$, and the move $h \to k$ is always a valid reduction.
\paragraph{Stranded Positions.} If $P_{\text{field}}\ne1$ but all heaps are $h < p^{n}-1$, the position is stranded. The player consolidates all heaps into a single new heap $H_{\text{new}} = C(P_{\text{field}})$, from which a winning move is then guaranteed.

\subsection{Conclusion: A Playable Model for the AES S-Box}
The Poly-MuM framework, played over $\GF(2^{8})$, provides a direct, game-based realisation of the algebra used in the AES S-box. Optimal play requires precisely the same operations---field multiplication and inversion modulo an irreducible polynomial---that underpin the S-box. Thus, Poly-MuM fulfils our original goal of translating sophisticated field-theoretic ideas into an engaging and intuitive gameplay mechanic.

\begin{table}[htbp!]
\centering
\caption{Multiplicative Inverses in $\GF(2^{3})$ with $I(x) = x^3+x+1$}
\label{tab:gf8_inverses}
\begin{tabular}{@{}cccccc@{}}
\toprule
\textbf{Polynomial} & \textbf{Integer} & \textbf{Power} & \textbf{Inverse} & \textbf{Inverse} & \textbf{Inverse} \\
\textbf{`s`} & \textbf{Rep.} & \textbf{of `x`} & \textbf{(Polynomial)} & \textbf{(Integer)} & \textbf{(Power)} \\
\midrule
$1$ & 1 & $x^0$ or $x^7$ & $1$ & 1 & $x^0$ or $x^7$ \\
$x$ & 2 & $x^1$ & $x^2+1$ & 5 & $x^6$ \\
$x^2$ & 4 & $x^2$ & $x^2+x+1$ & 7 & $x^5$ \\
$x+1$ & 3 & $x^3$ & $x^2+x$ & 6 & $x^4$ \\
$x^2+x$ & 6 & $x^4$ & $x+1$ & 3 & $x^3$ \\
$x^2+x+1$ & 7 & $x^5$ & $x^2$ & 4 & $x^2$ \\
$x^2+1$ & 5 & $x^6$ & $x$ & 2 & $x^1$ \\
\bottomrule
\end{tabular}
\end{table}

\clearpage
\section{Future Directions}\label{future-directions}

While extensive work has been done on nimbers, we believe there is
untapped potential in developing the theory of mumbers --- modular
analogs that operate under multiplicative invariants. The MuM game
framework provides a rich space to explore these ideas. Potential future directions include:
\begin{itemize}
\item Exploring mumber addition to construct a full modular field for gameplay.
\item Investigating other games that reduce to mumber-based Grundy-like strategies under multiplication mod p.
\item Visualizing mumber transitions using multi-dimensional geometric structures to gain insights.
\item Formalizing MuM's connection to or departure from Sprague-Grundy theory.
\end{itemize}

\section*{Statements and Declarations}
\textbf{Funding} None.\\
\textbf{Conflicts of Interest} None.\\
\textbf{Data Availability} All data are contained in this article.\\
\textbf{AI usage} AI language models (ChatGPT: OpenAI, Gemini: Google) were used for minor copy-editing; the authors are responsible for the content.

\bibliographystyle{plain}
\bibliography{references}

\clearpage
\appendix
\section{Appendix}\label{appendix}

\subsection{Recursive Mex Table}\label{recursive-mex-table}

The following table shows states in MuM\textsubscript{5} with 3 heaps and recursive computation of their Grundy values starting from a single heap.

\begin{sidewaystable}[htbp!]
\small
\begin{longtable}{|l|l|l|l|l|l|l|}
\caption{Recursive Mex Calculation for Sample States ($p=5$)} \label{tab:mex_example} \\
\hline
\makecell{\textbf{Heap 1}} &
\makecell{\textbf{Heap 2}} &
\makecell{\textbf{Heap 3}} &
\makecell{\textbf{Single Heap}\\\textbf{or Product}} &
\makecell{\textbf{Product}\\\textbf{Mod 5}} &
\makecell{
  \textbf{\(G\)-value}
} &
\makecell{\textbf{Recursive}\\\textbf{Mex}} \\
\hline
\endfirsthead
\multicolumn{7}{c}%
{{\bfseries \tablename\ \thetable{} -- continued from previous page}} \\
\hline
\makecell{\textbf{Heap 1}} &
\makecell{\textbf{Heap 2}} &
\makecell{\textbf{Heap 3}} &
\makecell{\textbf{Single Heap}\\\textbf{or Product}} &
\makecell{\textbf{Product}\\\textbf{Mod 5}} &
\makecell{
  \textbf{\(G\)-value}
} &
\makecell{\textbf{Recursive}\\\textbf{Mex}} \\
\hline
\endhead
\hline \multicolumn{7}{|r|}{{Continued on next page}} \\ \hline
\endfoot
\hline
\endlastfoot
\multicolumn{3}{|l|}{\{5,10,15,20,25,30,35,40,45,50,55,60\}} & 60 & 0 & 4 & 4 \\
\multicolumn{3}{|l|}{\{1,6,11,16,21,26,31,36,41,46,51,56,61\}} & 61 & 1 & 0 & 0 \\
\multicolumn{3}{|l|}{\{2,7,12,17,22,27,32,37,42,47,52,57,62\}} & 62 & 2 & 1 & 1 \\
\multicolumn{3}{|l|}{\{3,8,13,18,23,28,33,38,43,48,53,58,63\}} & 63 & 3 & 2 & 2 \\
\multicolumn{3}{|l|}{\{4,9,14,19,24,29,34,39,44,49,54,59,64\}} & 64 & 4 & 3 & 3 \\
\hline
1 & 1 & 6 & 6 & 1 & 0 & 0 \\
1 & 2 & 6 & 12 & 2 & 1 & 1 \\
1 & 3 & 6 & 18 & 3 & 2 & 2 \\
1 & 4 & 6 & 24 & 4 & 3 & 3 \\
1 & 6 & 6 & 36 & 1 & 0 & 0 \\
2 & 2 & 6 & 24 & 4 & 3 & 3 \\
2 & 3 & 6 & 36 & 1 & 0 & 0 \\
2 & 4 & 6 & 48 & 3 & 2 & 2 \\
2 & 6 & 6 & 72 & 2 & 1 & 1 \\
3 & 3 & 6 & 54 & 4 & 3 & 3 \\
3 & 4 & 6 & 72 & 2 & 1 & 1 \\
3 & 6 & 6 & 108 & 3 & 2 & 2 \\
4 & 4 & 6 & 96 & 1 & 0 & 0 \\
4 & 6 & 6 & 144 & 4 & 3 & 3 \\
6 & 6 & 6 & 216 & 1 & 0 & 0 \\
\hline
\end{longtable}
\end{sidewaystable}

\clearpage
\subsection{Decomposition of MuM\textsubscript{15}}\label{mod-3-mod-5-simplified-mumbers}

The following table shows states in MuM\textsubscript{15} and their decomposition into the `mod 3` and `mod 5` sub-games. The losing positions (L) are precisely those where the product is 1 in both sub-games.

\begin{sidewaystable}[htbp!]
\small
\begin{longtable}{|l|l|l|l|l|l|l|l|l|l|l|l|l|l|}
\caption{Decomposition of MuM\textsubscript{15} States} \label{tab:mum15_decomp} \\
\hline
H1 & M3 & M5 & H2 & M3 & M5 & H3 & M3 & M5 & \makecell{M3\\Prod} & \makecell{M5\\Prod} & Status \\
\hline
\endfirsthead
\multicolumn{13}{c}%
{{\bfseries \tablename\ \thetable{} -- continued from previous page}} \\
\hline
H1 & M3 & M5 & H2 & M3 & M5 & H3 & M3 & M5 & \makecell{M3\\Prod} & \makecell{M5\\Prod} & Status \\
\hline
\endhead
\hline \multicolumn{13}{|r|}{{Continued on next page}} \\ \hline
\endfoot
\hline
\endlastfoot
11 & 2 & 1 & 11 & 2 & 1 & 11 & 2 & 1 & 2 & 1 & W \\
11 & 2 & 1 & 11 & 2 & 1 & 13 & 1 & 3 & 1 & 3 & W \\
11 & 2 & 1 & 11 & 2 & 1 & 14 & 2 & 4 & 2 & 4 & W \\
11 & 2 & 1 & 11 & 2 & 1 & 16 & 1 & 1 & 1 & 1 & L \\
11 & 2 & 1 & 13 & 1 & 3 & 13 & 1 & 3 & 2 & 4 & W \\
11 & 2 & 1 & 13 & 1 & 3 & 14 & 2 & 4 & 1 & 2 & W \\
11 & 2 & 1 & 13 & 1 & 3 & 16 & 1 & 1 & 2 & 3 & W \\
11 & 2 & 1 & 14 & 2 & 4 & 14 & 2 & 4 & 2 & 1 & W \\
11 & 2 & 1 & 14 & 2 & 4 & 16 & 1 & 1 & 1 & 4 & W \\
11 & 2 & 1 & 16 & 1 & 1 & 16 & 1 & 1 & 2 & 1 & W \\
13 & 1 & 3 & 13 & 1 & 3 & 13 & 1 & 3 & 1 & 2 & W \\
13 & 1 & 3 & 13 & 1 & 3 & 14 & 2 & 4 & 2 & 1 & W \\
13 & 1 & 3 & 13 & 1 & 3 & 16 & 1 & 1 & 1 & 4 & W \\
13 & 1 & 3 & 14 & 2 & 4 & 14 & 2 & 4 & 1 & 3 & W \\
13 & 1 & 3 & 14 & 2 & 4 & 16 & 1 & 1 & 2 & 2 & W \\
13 & 1 & 3 & 16 & 1 & 1 & 16 & 1 & 1 & 1 & 3 & W \\
14 & 2 & 4 & 14 & 2 & 4 & 14 & 2 & 4 & 2 & 4 & W \\
14 & 2 & 4 & 14 & 2 & 4 & 16 & 1 & 1 & 1 & 1 & L \\
14 & 2 & 4 & 16 & 1 & 1 & 16 & 1 & 1 & 2 & 4 & W \\
16 & 1 & 1 & 16 & 1 & 1 & 16 & 1 & 1 & 1 & 1 & L \\
\hline
\end{longtable}
\end{sidewaystable}
\end{document}